\documentclass[oneside,english]{amsart}
\usepackage[T1]{fontenc}
\usepackage[latin9]{inputenc}
\usepackage{geometry}
\usepackage{amstext}
\usepackage{amsthm}
\usepackage{amssymb}
\usepackage{amsmath}
\usepackage{hyperref}

\makeatletter
\numberwithin{equation}{section}
\numberwithin{figure}{section}

\theoremstyle{plain}
\newtheorem{thm}{Theorem}

\newtheorem{lem}[thm]{Lemma}
\newtheorem{cor}[thm]{Corollary}

\theoremstyle{definition}
\newtheorem{defn}[thm]{Definition}

\newtheorem{conj}[thm]{Conjecture}

\theoremstyle{remark}

\makeatother

\usepackage{babel}
\usepackage{tikz}
\usepackage{bbm}

\usepackage{algorithm}
\usepackage[noend]{algpseudocode}

    \algnewcommand\algorithmicto{\textbf{to}}
    \algnewcommand\To{\algorithmicto{} }
    \algnewcommand\algorithmicswitch{\textbf{switch}}
    \algnewcommand\algorithmiccase{\textbf{case}}
    \algdef{SE}[SWITCH]{Switch}{EndSwitch}[1]{\algorithmicswitch\ #1\ \algorithmicdo}{\algorithmicend\ \algorithmicswitch}%
    \algdef{SE}[CASE]{Case}{EndCase}[1]{\algorithmiccase\ #1}{\algorithmicend\ \algorithmiccase}%
    \algtext*{EndSwitch}%
    \algtext*{EndCase}%

\newcommand{\eps}{\varepsilon}

\newcommand{\Ber}{\mathsf{Ber}}

\newcommand{\blfootnote}[1]{{\renewcommand{\thefootnote}{\roman{footnote}}\footnotetext[0]{#1}}}

\begin{document}

\title{Iterated Entropy Derivatives and Binary Entropy Inequalities}
\author{Tanay Wakhare${}^{1}$}

\begin{abstract}
We embark on a systematic study of the $(k+1)$-th derivative of $x^{k-r}H(x^r)$, where $H(x):=-x\log x-(1-x)\log(1-x)$ is the binary entropy and $k\geq r\geq 1$ are integers. Our motivation is the conjectural entropy inequality $\alpha_k H(x^k)\geq x^{k-1}H(x)$, where $0<\alpha_k<1$  is given by a functional equation. The $k=2$ case was the key technical tool driving recent breakthroughs on the union-closed sets conjecture. We express $ \frac{d^{k+1}}{dx^{k+1}}x^{k-r}H(x^r)$ as a rational function, an infinite series, and a sum over generalized Stirling numbers. This allows us to reduce the proof of the entropy inequality for real $k$ to showing that an associated polynomial has only two real roots in the interval $(0,1)$, which also allows us to prove the inequality for fractional exponents such as $k=3/2$. The proof suggests a new framework for proving tight inequalities for the sum of polynomials times the logarithms of polynomials, which converts the inequality into a statement about the real roots of a simpler associated polynomial.
\end{abstract}
\maketitle

\blfootnote{$^{1}$~Department of Electrical Engineering and Computer Science, MIT, Cambridge, MA 02139, USA}
\blfootnote{ \quad   \email{twakhare@mit.edu}}

\section{Introduction}

The union-closed sets conjecture is a notorious open problem, stating that any set family $\mathcal{F}\subseteq 2^{[n]}$ which is \textit{union-closed} (so that the union of two sets in $\mathcal{F}$ is also in the system) contains a "popular" element of the ground set contained in at least a $1/2$ fraction of the sets of $\mathcal{F}$. 

Although the conjecture is still unproven, Gilmer made a recent breakthrough stating that any union-closed set system contains an element in at least an $0.01$ fraction of the sets in $\mathcal{F}$. The constant $0.01$ was quickly improved to $\frac{3-\sqrt{5}}{2} \approx 0.38197$ \cite{AHS22, Sawin22, CL22}. Building upon more sophisticated coupling arguments suggested by \cite{Sawin22,Yu23,Cambie22}, the current best constant is $\approx 0.38237$ \cite{liu2024improving}, though the method suffers natural limitations. The survey \cite{Cambie23} summarizes recent progress and barriers, but new ideas will be needed to prove the full union-closed sets conjecture. 



The $k=2$ case of the following inequality \eqref{mainineq} was conjectured by Gilmer, and was one of the key technical tools underlying his breakthrough. This case was proved using computer calculations by \cite{AHS22}. Studying the extension to approximate $k$-union closed set systems led to \cite{Yuster23} conjecturing inequality \eqref{mainineq} for every integer $k\geq 2$ and proving it for $k=3,4$. It later emerged that Boppana proved the $k=2$ case several decades earlier \cite{Boppana85}. He recently republished a simplified proof \cite{Boppana23}, which forms the basis for Sawin's \cite{Sawin22} improvements and is the proof we build upon. The main contribution of this paper is to reduce the proof of the real $k\geq 1$ case to a conjecture about the roots of an explicit polynomial, which suggests a general framework to prove tight inequalities involving the sum of logarithms of polynomials.

\begin{conj}\label{mainthm} Let $k\geq 1$ be real and $ 0<\alpha_k<1$ be the unique solution of 
\begin{equation}\label{alphafunc}
    \alpha_k = \frac{1}{(1+\alpha_k)^{k-1}}
\end{equation}
in $(0,1)$. Then 
\begin{equation}\label{mainineq}
    \alpha_k H(x^k) \geq x^{k-1}H(x), \quad 0\leq x\leq 1,
\end{equation}
where $H(x):=-x\log x - (1-x)\log (1-x)$ is the binary entropy. We have equality at $x = 0,\frac{1}{1+\alpha_k},1$. 
\end{conj} 

Throughout, all logarithms are to base $e$. Lemma \ref{lem:alphabound} shows that the functional equation \eqref{alphafunc} has a unique solution satisfying $1/k <\alpha_k<1$, and Lemma \ref{lem:alphak-estimate} shows that $\alpha_k = \frac{\log k}{k} + O_k \left( \frac{\log \log k}{k}\right)$ asymptotically for large $k$.

The natural transformation for this problem is $x = \frac{1}{1+y}$, since we now study this inequality over $y$ in $(0,\infty)$ instead of over $x$ in $(0,1)$, which maps the root at $x = \frac{1}{1+\alpha_k}$ to a root at $y=\alpha_k$. Writing $x_k=\frac{1}{1+\alpha_k}$, this functional equation is equivalent to $x_k+x_k^k=1$. The equation $x+x^k=1$ corresponds to the characteristic function of Fibonacci type recurrences like $F_{n}=F_{n-1}+F_{n-k}$. This explains the appearance of the golden ratio in the $k=2$ case studied for the union closed sets conjecture, since $x+x^2=1$ has roots closely connected to the golden ratio. This also motivates studying $x_k,\alpha_k$ in terms of generalized Fibonacci polynomials, related to \cite{Cigler22}.

We show that Conjecture \ref{rootconj} implies Conjecture \ref{mainthm}. This is a strong statement about polynomial roots, since the following polynomial $p_{k,r}(x)$ has degree $k^2+kr-r$, but we conjecture it to only have two roots in $(0,1)$. The following conjecture also allows us to rigorously prove Conjecture \ref{mainthm} for any rational exponent using a finite calculation, such as for $k=3/2$.

\begin{conj}\label{rootconj}Let $k>r\geq 1$ be integers. Define the \textbf{entropy polynomial} 
\begin{equation}
h_{k,r}(x):=\sum_{j=0}^{k-1} x^{rj}\sum_{v=0}^j  \frac{(-1)^{j-v}}{v+1} \binom{rv+k}{k}  \binom{k}{j-v}
\end{equation}
and let $\alpha_k$ satisfy the functional equation \eqref{alphafunc}. Then the polynomial 
\begin{equation}
   p_{k,r}(x):=\alpha_{k/r} k(1-x^r)^{k} h_{k,k}(x) - r(1-x^k)^{k}h_{k,r}(x)
\end{equation}
has exactly two real roots in $(0,1)$, counting multiplicity. 
\end{conj}
Note that Lemma \ref{lem:alphabound} states that $\alpha_{k/r}\frac{k}{r}>1$, so that the first polynomial is multiplied by a larger constant factor. 

\begin{thm}\label{thm:reduction}
If Conjecture \ref{rootconj} holds for a particular $k>r$ pair, then inequality \eqref{mainineq} holds for the exponent $k/r$. If Conjecture \ref{rootconj} holds for all coprime $k>r\geq 1$, then inequality \eqref{mainineq} holds for all real $k\geq 1$.
\end{thm}

For instance, a quick calculation shows that Conjecture \ref{rootconj} holds for $k=3,r=2$. A natural approach is to use a special case of Descartes' rules of signs, which states that if a polynomial has two coefficient sign changes, then it has either $0$ or $2$ positive real roots. Numerically, under the change of variables $x = \frac{1}{1+y}$, the polynomial $(1+y)^{k^2+kr-r} p_{k,r}\left(\frac{1}{1+y}\right)$ always has two sign changes. Since $h_{k,r}$ has degree $kr-r$ and $p_{k,r}$ has degree $k^2+kr-r$, the factor of $(1+y)^{k^2+kr-r}$ ensures that the resulting expression is a polynomial while only introducing extra roots at $y=-1$. If this has at most two real roots for $y$ in $(0,\infty)$, these correspond to at most two real roots of $p_{k,r}(x)=p_{k,r}\left(\frac{1}{1+y}\right)$ in $(0,1)$. However, the coefficients in $y$ become unwieldy double or triple sums, from which it is difficult to deduce the sign pattern.


Some example cases are
\begin{align*}
    h_{1,1}(x) = 1, \quad\quad h_{2,2}(x) = 1+x^2, \quad\quad h_{3,3}(x) = 1+7x^3+x^6,\quad \quad h_{4,4}(x) = 1+31x^4+31x^8+x^{12}.
\end{align*}
and
\begin{align*}
    h_{4,1}(x) &= 1-\frac{3}{2}x+x^2-\frac{1}{4}x^3, \quad\quad\quad\thinspace\thinspace h_{4,2}(x) = 1+\frac{7}{2}x^2 -\frac23 x^4 +\frac16 x^6, \\
    h_{4,3}(x) &= 1+\frac{27}{2}x^3+6x^6-\frac14 x^9, \quad\quad h_{4,4}(x) = 1+31x^4+31x^8+x^{12}.
\end{align*}

This motivates the study of the binomial sums
\begin{align}\label{hkrjdef}
h_{k,r,j}:=\sum_{v=0}^j  \frac{(-1)^{j-v}}{v+1} \binom{rv+k}{k}  \binom{k}{j-v}     
\end{align}
for all values of the parameters $k,r,j$. While this is rational valued in general, several rescaled integer valued multiples of $h_{k,r,j}$ do not appear in the OEIS. For instance, using a variation of the proof of Lemma \ref{thm:finite-diff} using finite difference operators, we can show that $h_{k,r,k} = \frac{1}{k+1}\binom{r-1}{k}$ for all $r,k\geq 1$, which is $0$ for $r\leq k$. An interesting and related open problem is computing a simple representation for $h_{k,r}(x)$ under the change of variables $x\mapsto 1-x$ or $x^r \mapsto 1-x^r$, mirroring the symmetry of the binary entropy $H(x) = H(1-x)$.

Our key technical tool is several equivalent expansions for the $(k+1)$-st derivative of $x^{k-r}H(x^r)$, which are all functions of $x^r$. The first expansion expresses the derivative as a single infinite series, the second factors out a single root at $0$ and a root of multiplicity $k$ at $1$, which leaves the numerator as a polynomial. The last generalizes and simplifies \cite[Lemma 3.5, Lemma 3.8]{Yuster23} and rewrites the $(k+1)$-st derivative in terms of a different rational basis, with coefficients given by generalized Stirling numbers.
\begin{thm}\label{thm3}Let $S(k,\ell|\alpha,\beta,\gamma)$ denote the generalized Stirling numbers of Hsu and Shiue, defined in Equation \eqref{stirdef2}. Let $k\geq r\geq 1$ be positive integers. For $0<x<1$ we have
\begin{align}
 \left(\frac{d}{dx}\right)^{k+1}x^{k-r}H(x^r) &= -r\cdot k! \sum_{\ell=0}^\infty \binom{k+r\ell}{k} \frac{1}{\ell+1} x^{r\ell-1}\label{thm3eq1} \\
&= -\frac{r\cdot k!}{x(1-x^r)^{k}} \sum_{j=0}^{k-1} x^{rj}\sum_{v=0}^j\frac{(-1)^{j-v}}{v+1}\binom{rv+k}{k} \binom{k}{j-v} \label{mainthmfullstatement}\\
&=  -\sum_{\ell=0}^{k-1} \ell!S(k,\ell+1|1,r,k-r)r^{\ell+2} \frac{x^{r\ell-1}}{(1-x^r)^{\ell+1}}.
 \end{align}
\end{thm}

\begin{cor}\label{cor6}The special case $r=1$ satisfies
\begin{align}
     \left(\frac{d}{dx}\right)^{k+1}x^{k-1}H(x) = \frac{(k-1)!}{x^2} \left(1 - \frac{1}{(1-x)^k}\right).
 \end{align}
\end{cor}

\begin{cor}\label{cor7}
The special case $r=k$ has the following additional simplifications in terms of $s$-binomial coefficients defined in Definition \eqref{snomialdef}, where $\omega=e^{\frac{2\pi i}{k}} $ is a primitive $k$-th root of unity:
\begin{align}
 \left(\frac{d}{dx}\right)^{k+1}H(x^k) &= -\frac{k!}{x} \sum_{j=0}^{k-1} \frac{1}{(1-\omega^jx)^k} \\
 &=-\frac{k\cdot k!}{x(1-x^k)^k} \sum_{\ell=0}^{k-1} \binom{k}{\ell k}_{k-1}x^{k\ell} \\
 &= -k\cdot k! \sum_{\ell=0}^\infty \binom{k+k\ell-1}{k-1} x^{k\ell-1} .
no \end{align}
\end{cor}
The scaling $r v$ in the inner binomial coefficient $\binom{rv+k}{k}$ in Equation \eqref{mainthmfullstatement} is what makes the analysis here difficult. One common classical tool to deal with the $rv$ scaling is the Rothe-Hagen identity \cite[Table 202]{Knuth94} and its generalizations, for example due to Gould \cite{Gould61}. However, the Rothe-Hagen identity contains binomials of the form $\binom{rv+k}{v}$, where $k,r$ are parameters and $v$ is the summation index. The Lagrange inversion formula also yields series with binomial coefficients $\binom{rv+k}{v}$, such as the expression for $\frac{1}{1+\alpha_k}$ from Lemma \ref{lem:lagrange}. Instead, we require binomials of the form $\binom{rv+k}{k}$. We could also compute a Fourier expansion by writing the sum over all $v$ instead of $rv$, and then inserting $\frac{1}{r}\sum_{j=0}^{r-1}\omega^{jv}$, where $\omega$ is a primitive $r$-th root of unity. However, this only provides a simplification in the case $k=r$, where it is used in the proof of Corollary \ref{cor7}.

Inequality \eqref{mainineq} also has an information theoretic interpretation. Letting $X_1,\ldots,X_k\sim \Ber(x)$ be Bernoulli distributed bits and $A_j:=\wedge_{i=1}^j X_i$ denote the binary AND of the first $j$ bits, we have $H(x^k) = H(A_k)$ and $ H(A_k|A_{k-1}) = x^{k-1}H(x) $ the conditional entropy of the $k$-th bit. This gives a strong data processing inequality comparing the entropy of the AND of $k$ bits to the entropy of the AND of $k$-th bit conditioned on the AND of the previous $k-1$ bits.

In Section \ref{sec:def} we introduce the generalized Stirling numbers of Hsu and Shiue, and demonstrate a connection to generalized Bernoulli and Eulerian numbers. In Section \ref{sec:findiff} we evaluate $h_{k,r,j}$ in certain regimes of $k,r,j$. In Section \ref{sec3} we prove the entropy expansions of Theorem \ref{thm3}. In Section \ref{sec4} we prove Corollaries \ref{cor6} and \ref{cor7}, which are the $r=1$ and $r=k$ cases of Theorem \ref{thm3}. In Section \ref{sec:reduction} we prove Theorem \ref{thm:reduction}, an equivalence between our main entropy inequality and counting real roots of $h_{k,r}(x)$ in $(0,1)$. Finally, in Section \ref{sec:alpha} we study the scaling constant $\alpha_k$.

The proof suggests a more general framework for proving tight inequalities for logarithms of polynomials of the form
$$ f(x):=\sum_{i} p_i(x)\log(1-q_i(x)) \geq 0,$$
where both $p_i(x)$ and $q_i(x)$ are polynomials in $x$. Such functions arise as free energies in problems in statistical mechanics or in constraint satisfaction problems, such as in the study of graph $k$-coloring \cite[Equation (8)]{Coja2013} or boolean $k$-SAT \cite[Equation (3)]{Mertens2006}. First, we manually find the roots of $f(x)$. Next, we pass to the $(n+1)$-st derivative, where $n$ is the maximum degree $\max_i (\deg  p_i(x)+\deg q_i(x))$. Taking this number of derivatives leads to a rational function. We then use classical methods to identify the number of roots of the $(n+1)$-st derivative, which is a rational function and more tractable than the original logarithmic $f(x)$. We then appeal to Rolle's theorem in the form that a function can have at most one more root than its derivative on a given interval. We use Rolle's theorem to pass from the $(n+1)$-st derivative to the original function, which can have at most $(n+1)$ more roots than the $(n+1)$-st derivative. If we are lucky, then in the first step we identified all roots of $f(x)$. Finally, we show that $f(x)$ takes positive values between each root, which implies that it is positive everywhere. The innovation over previous methods is that if we can control the \textit{multiplicities} of the roots of $f(x)$ and the \textit{multiplicities} of the roots of the $(n+1)$-st derivative, which is now a polynomial, Rolle's theorem allows us to deduce the desired inequality.

\section{Definitions}\label{sec:def}

The ubiquitous \textbf{Stirling numbers of the second kind} are defined as the solutions to the recurrence \cite[Equation (6.3)]{Knuth94} and have the closed form \cite[Equation (6.19)]{Knuth94}:
\begin{align}
     S(n+1,\ell) &= S(n,\ell-1)+\ell S(n,\ell)\label{stirdef5}, \\
     \ell!S(n,\ell) &= \sum_{v=0}^\ell(-1)^{\ell-v}\binom{\ell}{v} v^n . \label{stirdef6}
\end{align}

Define the \textbf{scaled Pochhammer} symbol $(z|\alpha)_n:= z(z-\alpha)\cdots (z-(n-1)\alpha), n\geq 1$. Hsu and Shiue \cite{Hsu88} introduced generalized Stirling numbers using these scaled Pochhammer symbols. Let $\alpha,\beta,\gamma\in \mathbb{R}$. Define \textbf{generalized Stirling numbers} $S(n,\ell|\alpha,\beta,\gamma)$ via the change of basis relation
\begin{equation}
    (z|\alpha)_n = \sum_{\ell=0}^n S(n,\ell|\alpha,\beta,\gamma) (z-\gamma|\beta)_n
\end{equation}
and initial conditions $S(0,0|\alpha,\beta,\gamma) = 1, S(n,0|\alpha,\beta,\gamma) = (\gamma|\alpha)_n$.

Many properties of these, discovered by subsequent researchers, are surveyed in the book \cite{Mansour16}. We will require the following recurrence and closed form \cite[Theorem (4.51), Theorem (4.52)]{Mansour16}:
\begin{align}
    S(n+1,\ell|\alpha,\beta,\gamma) &= S(n,\ell-1|\alpha,\beta,\gamma) + (\ell\beta-n\alpha+\gamma)S(n,\ell|\alpha,\beta,\gamma),\label{stirdef2} \\
    S(n,\ell|\alpha,\beta,\gamma) &= \frac{(-1)^\ell}{\beta^\ell \ell!} \sum_{j=0}^\ell(-1)^j \binom{\ell}{j}(\beta j +\gamma|\alpha)_n. \label{stirdef3}
\end{align}
In the remainder of this paper, we often take $\alpha=1$ in the definition of the generalized Stirling numbers, which gives
\begin{equation}\label{stirdef7}
    \ell! S(n,\ell|1,\beta,\gamma) \beta^\ell = \sum_{v=0}^\ell (-1)^{\ell-v} \binom{\ell}{v}n!\binom{\beta v+\gamma}{n}.
\end{equation}

We also require the Dobi\'nski type formula \cite[Equation (27)]{Hsu88} which factors $e^x$ out of the series:
\begin{equation}\label{stirdef4}
    \sum_{\ell=0}^\infty \frac{x^\ell}{\ell!} \binom{r\ell+s}{n} =  \frac{1}{n!} \sum_{\ell=0}^n S(n,\ell|1,r,s)r^\ell e^xx^\ell.
\end{equation}

Note that by comparing recurrences and initial conditions, we can show that 
$$ C(k,t,j) = \frac{k^j}{(k-1)!}S(t,j|1,k,0),$$
where $C(k,t,j)$ are the rational coefficients introduced by Yuster \cite{Yuster23} in his work on the approximate $k$-union closed sets conjecture, which considers the $r=1,k$ cases of our result. Furthermore, the $C(k,t,j)$ coefficients have been studied before and are exactly the generalized factorial coefficients of \cite[Definition 8.2]{Charalambides02}, since they satisfy the same recurrence and initial conditions.

\subsection{Classical analogy}
Comparing the two closed form expressions for Stirling numbers \eqref{stirdef6} and \eqref{stirdef7} shows that, up to normalization by $\beta$, we replace the term $v^n$ with $$n!\binom{\beta v+\gamma}{n} =(\beta v +\gamma)(\beta v +\gamma-1)\cdots (\beta v +\gamma-n+1).$$ 

If we further specialize $\gamma=0$, we can take the limit
\begin{align}
\lim_{\beta\to \infty} \frac{S(n,\ell|1,\beta,0)\beta^\ell}{\beta^n} &= \lim_{\beta\to \infty} \frac{1}{\ell! \beta^n} \sum_{v=0}^\ell (-1)^{\ell-v} \binom{\ell}{v}n!\binom{\beta v}{n} \\
&= \frac{1}{\ell!} \sum_{v=0}^\ell (-1)^{\ell-v} \binom{\ell}{v}v^n \\
&= S(n,\ell).
\end{align}

The \textbf{Eulerian numbers} $A_{n,k}$ have the closed form \cite[Equation (6.38)]{Knuth94}
\begin{align}
    A_{n,\ell} &= \sum_{v=0}^{\ell}(-1)^{\ell-v}\binom{n+1}{\ell-v}(v+1)^n.  \label{eulerdef2}
\end{align} We introduce the companion sequence of \textbf{generalized Eulerian numbers}
\begin{equation}\label{eulerdef3}
    A_{n,\ell}^{(r,s)} = n!\sum_{v=0}^{\ell}(-1)^{\ell-v} \binom{n+1}{\ell-v} \binom{(v+1)r+s}{n}.
\end{equation}
By a similar argument to the Stirling case we see that $$ \lim_{r\to \infty}  \frac{A_{n,\ell}^{(r,0)}}{r^n} = A_{n,\ell}.$$ Note that if $\gamma\neq 0$ or $s\neq 0$ we do not reduce to standard Eulerian and Stirling numbers in the large $r$ limit. 

There are many classical relations linking Stirling numbers, sum of powers, Eulerian numbers, and Bernoulli numbers. One of the key identities linking moments, Stirling numbers, and Eulerian numbers \cite[Equation (7.46)]{Knuth94} is
\begin{equation}
\left(z \frac{d}{dz}\right)^n \frac{1}{1-z} = \sum_{\ell=1}^\infty \ell^n z^\ell= \sum_{j=0}^n j!S(n,j) \frac{z^j}{(1-z)^{j+1}} = \frac{z}{(1-z)^{n+1}} \sum_{j=0}^n A_{n,j}z^j.   
\end{equation}
Combining Lemma \ref{lem:stir1eq1} with some further calculations gives a generalization of this transformation involving the parameters $r,s$:
\begin{align}
 n!\sum_{\ell=1}^\infty \binom{r\ell+s}{n} z^\ell = \sum_{j=0}^n j!S(n,j|1,r,s)r^j \frac{z^j}{(1-z)^{j+1}} = \frac{z}{(1-z)^{n+1}}\sum_{j=0}^n A_{n,j}^{(r,s)}z^j .
\end{align}
The existence of this transformation is what leads to the definition of $A_{n,j}^{(r,s)}$. These definitions of generalized Stirling and Eulerian numbers suggest an analog of the classical calculus where we systematically replace powers $v^t$ with the scaled Pochhammer $t!\binom{\beta v}{t} = (\beta  v)(\beta  v-1)\cdots (\beta v-t+1)$. This introduces the free parameter $\beta$, and we can recover all of the classical sequences by normalizing and taking the $\beta \to \infty$ limit.

Beginning with the closed form for Bernoulli numbers \cite[Equation (24.6.9)]{DLMF}
\begin{equation}
    B_n = \sum_{\ell=0}^n \sum_{v=0}^\ell \frac{(-1)^v}{\ell+1} \binom{\ell}{v} {v^n},
\end{equation}
we then define \textbf{generalized Bernoulli numbers} by the closed form
\begin{equation}
        B_n^{(r,s)}:=n!\sum_{\ell=0}^n \sum_{v=0}^\ell \frac{(-1)^v}{\ell+1} \binom{\ell}{v} \binom{r v+s}{n}.
\end{equation}
We leave it as an open question to explore the links between these generalizations of Bernoulli, Eulerian, and Stirling numbers.

\section{Finite differences}\label{sec:findiff}

We will need to understand $h_{k,r,j}$ in more detail, in particular its vanishing for $j\geq k \geq r \geq 1$. Using finite difference operators, we can provide an expression for $h_{k,r,j}$ which sums over $k-j$ terms instead of $j$ terms. This gives explicit expressions for the leading coefficients of $h_{k,r}(x)$, since for example we have $h_{k,r,k-1} = (-1)^r \frac{1}{\binom{k}{r}} $ and $(-1)^k h_{k,r,k-2} = (-1)^{r+1}\frac{k}{\binom{k}{r}}+\binom{k-2r}{k} $. Note the appearance of binomial coefficients with a negative upper index, defined as $\binom{-n}{k}=(-1)^k\binom{n+k-1}{k}$ for $k\geq 0$ and any integer $n$ \cite[Equation (1.2.6)]{DLMF}.
\begin{lem}\label{thm:finite-diff} 
Consider integer $k\geq r \geq 1$. If $j\geq k$, then $h_{k,r,j} = 0$. If $1\leq j<k$, then
\begin{equation}
    h_{k,r,j}=(-1)^{j+r+1} \frac{\binom{k}{j+1}}{\binom{k}{r}} + \sum_{v=2}^{k-j} \frac{(-1)^{j+v}}{v-1} \binom{k}{j+v}\binom{k-rv}{k}.
\end{equation}
\end{lem}
\begin{proof}
    Define the polynomial $q(x) =\frac{1}{x+1}\binom{rx+k}{k}$. For $r\leq k$ the numerator contains the factor $rx+r$ which cancels with $x+1$ in the denominator, so that $q(x)$ is a polynomial of degree $\leq k-1$. We apply the finite difference operator $\Delta$ defined by $\Delta q(x) := q(x+1)-q(x)$, which lowers the degree of a polynomial $q(x)$ by $1$. Therefore, we have the key identity
    \begin{equation}\label{eq:finite-diff}
        \Delta^k q(x) = \sum_{v=0}^k (-1)^{k-v} \binom{k}{v} q(x+v) =  \sum_{v=0}^k (-1)^{k-v} \binom{k}{v} \frac{1}{x+v+1} \binom{rx+rv+k}{k} = 0,
    \end{equation}
    which is $0$ because we have lowered the degree of a degree $k-1$ polynomial $k$ times. 

    Consider the case $j\geq k\geq r\geq 1$. We rewrite the definition of $h_{k,r,j}$ as 
    \begin{align*}
        h_{k,r,j} &= \sum_{v=0}^j  \frac{(-1)^{j-v}}{v+1} \binom{rv+k}{k}  \binom{k}{j-v} \\
         &= \sum_{v=j-k}^j  \frac{(-1)^{j-v}}{v+1} \binom{rv+k}{k}  \binom{k}{j-v} \\
         &=  \sum_{v=0}^k  \frac{(-1)^{k-v}}{j-k+v+1} \binom{rv+k + r(j-k)}{k}  \binom{k}{k-v} \\
        &=  \sum_{v=0}^k  \frac{(-1)^{k-v}}{j-k+v+1} \binom{rv+k + r(j-k)}{k}  \binom{k}{v}.
    \end{align*}
    We first truncated the sum from $v=j-k\geq 0$ to $v=k$ since $\binom{k}{j-v}$ vanishes outside this range. We then shifted the $v$ sum by $j-k$ and used the symmetry $\binom{k}{k-v} = \binom{k}{v}$. Comparing against Equation \eqref{eq:finite-diff}, we see that this is exactly $\Delta^k q(x)\vert_{x=j-k} =0$, since setting $x=j-k\geq 0$ does not lead to any singular terms. 
    
    Now consider the case $k\geq r\geq 1$ and $k>j\geq 1$. Consider the limit $x\to -j$ in Equation \eqref{eq:finite-diff}, so that the only singular term is at $v = j-1$, where we have
    \begin{align*}
        \lim_{x\to -j}(-1)^{k-j-1} &\binom{k}{j-1} \frac{1}{x+j} \binom{rx+rj+k-r}{k} = (-1)^{k-j-1} \binom{k}{j-1}\lim_{x\to 0} \frac{1}{x} \binom{rx+k-r}{k} \\
        &= (-1)^{k-j-1} \binom{k}{j-1}\frac{r}{k!} \lim_{x\to 0} \frac{(rx+k-r)(rx+k-r-1)\cdots (rx-r+1)}{rx} \\
        &= (-1)^{k-j-1} \binom{k}{j-1} \frac{r}{k!}\cdot (k-r)! (r-1)! (-1)^{r-1} \\
        &= (-1)^{k-j-r} \frac{\binom{k}{j-1}}{\binom{k}{r}}.
    \end{align*}
    The $rx$ terms in the numerator and denominator cancelled, so that substituting $x=0$ was a well defined operation. 
    Hence the finite difference result \eqref{eq:finite-diff} reduces to
    \begin{align}
        0=\Delta^k q(x)\big\vert_{x=-j} &= (-1)^{k-j-r} \frac{\binom{k}{j-1}}{\binom{k}{r}} + \sum_{v=0}^{j-2} \frac{(-1)^{k-v}}{v-j+1} \binom{k}{v}  \binom{rv-rj+k}{k} +  \sum_{v=j}^{k} \frac{(-1)^{k-v}}{v-j+1} \binom{k}{v}  \binom{rv-rj+k}{k} \nonumber \\
        &= (-1)^{k-j-r} \frac{\binom{k}{j-1}}{\binom{k}{r}} + h_{k,r,k-j} + \sum_{v=0}^{j-2} \frac{(-1)^{k-v}}{v-j+1}  \binom{k}{v} \binom{rv-rj+k}{k}. \label{reciprocity1}
    \end{align}
    
    Here, we rewrote $h_{k,r,k-j}$ as
    \begin{align*}
        h_{k,r,k-j}:&=\sum_{v=0}^{k-j}  \frac{(-1)^{k-j-v}}{v+1} \binom{rv+k}{k}  \binom{k}{k-j-v} \\
        &=\sum_{v=j}^k  \frac{(-1)^{k-v}}{v-j+1} \binom{rv-rj+k}{k}  \binom{k}{k-v} \\
        &=\sum_{v=j}^k  \frac{(-1)^{k-v}}{v-j+1} \binom{rv-rj+k}{k}  \binom{k}{v},
    \end{align*}
    where we shifted the summation index $v$ by $j$, reversed the order of summation, and used the symmetry $\binom{k}{k-v} = \binom{k}{v}$.

    We also have
    \begin{align*}
        \sum_{v=0}^{j-2}\frac{ (-1)^{k-v}}{v-j+1} \binom{k}{v}  \binom{rv-rj+k}{k} 
        &= \sum_{v=0}^{j-2} \frac{(-1)^{k+j-v} }{-(v+1)} \binom{k}{j-v-2}  \binom{-rv-2r+k}{k}  \\
        &= \sum_{v=2}^{j} \frac{(-1)^{k+j-v+1}}{v-1} \binom{k}{j-v}  \binom{k-rv}{k},
    \end{align*}
    where we reversed the order of summation and then shifted the summation over $v$ by $2$. Finally, reversing $j \mapsto k-j$ in Equation \eqref{reciprocity1} and noting $\binom{k}{k-j-v} = \binom{k}{j+v}$ completes the proof.
\end{proof}

\section{Entropy derivative closed forms}\label{sec3}
This section proves the closed forms for iterated entropy derivatives from Theorem \ref{thm3}. We begin with a fundamental infinite series expansion for the binary entropy. The key is that we consider $x^k\log x$ as an analytic function around $0$, which we do not series expand, while we series expand $\log(1-x^k)$.

\begin{lem}\label{lem:fund}
    For integer $k \geq 1$ and real $0\leq x\leq 1$ we have the expansion 
\begin{equation}
    H(x^k) = -x^k\log x^k -(1-x^k)\log(1-x^k) = -kx^k\log x + x^k -\sum_{\ell=1}^\infty \frac{x^{k(\ell+1)}}{\ell(\ell+1)}.
\end{equation}
\end{lem}
\begin{proof}
Recall that $-\log(1-x) = \sum_{\ell=1}^\infty \frac{x^\ell}{\ell}$ for $0<x<1$.
Consider the series 
\begin{align*}
    \sum_{\ell=1}^\infty \frac{x^{\ell+1}}{\ell(\ell+1)} &= \sum_{\ell=1}^\infty {x^{\ell+1}} \left(\frac{1}{\ell}-\frac{1}{\ell+1}\right) \\
    &= x \sum_{\ell=1}^\infty \frac{x^\ell}{\ell} -  \sum_{\ell=2}^\infty \frac{x^\ell}{\ell} \\
    &= -x \log(1-x) + (x+\log(1-x)) \\
    &= x + (1-x)\log(1-x).
\end{align*}
Mapping $x\mapsto x^k$ and substituting into the definition of $H(x^k) = -x^k\log x^k - (1-x^k)\log(1-x^k)$ finishes the proof.

Note that at $x=0$ this approaches $H(0)=0$ and at $x=1$ we can telescope $$\sum_{\ell=1}^\infty \frac{1}{\ell(\ell+1)} = \sum_{\ell=1}^\infty  \left(\frac{1}{\ell}-\frac{1}{\ell+1}\right) = 1,$$
so that the series converges for $0\leq x\leq 1$.
\end{proof}

We now differentiate termwise to obtain an expression for the $(k+1)$-st derivative.
\begin{lem}\label{lem8}
        For integer $k \geq r \geq 1$ and real $0< x< 1$ we have \begin{align}
 \left(\frac{d}{dx}\right)^{k+1}x^{k-r}H(x^r) &= -r\cdot k! \sum_{\ell=0}^\infty \binom{k+r\ell}{k} \frac{1}{\ell+1} x^{r\ell-1}.
\end{align}
\end{lem}
\begin{proof}
We begin with Lemma \ref{lem:fund} in the form 
\begin{equation}\label{lem:fund2}
x^{k-r}H(x^r) = -rx^k\log x+x^k - \sum_{\ell=1}^\infty \frac{x^{k+r\ell}}{\ell(\ell+1)}.
\end{equation}
Note that $$\left(\frac{d}{dx}\right)^{k+1}x^k\log x = \frac{k!}{x}$$ and $$\left(\frac{d}{dx}\right)^{k+1} x^{k+r\ell} = \frac{(k+r\ell)!}{(r\ell-1)!}x^{r\ell-1},$$
so that after differentiating $(k+1)$ times termwise we have
\begin{align}
    \left(\frac{d}{dx}\right)^{k+1}x^{k-r}H(x^r) = -\frac{r\cdot k!}{x} - \sum_{\ell=1}^\infty \frac{(k+r\ell)!}{(r\ell-1)!(\ell)(\ell+1)}x^{r\ell-1}. \label{gencase1}
\end{align}
Rewrite the factorials as 
$$ \frac{(k+r\ell)!}{(r\ell-1)!(\ell)(\ell+1)} = r\cdot k!\frac{(k+r\ell)!}{(r\ell-1)!k!(r\ell)(\ell+1)} = r\cdot k! \binom{k+r\ell}{k} \frac{1}{\ell+1},$$
and then recognize $-\frac{r\cdot k!}{x} $ as the $\ell=0$ term of the sum. The key observation is that the factorial ratio cancels nontrivially. Finally, the sum in Equation \eqref{gencase1} becomes 
\begin{align*}
 \left(\frac{d}{dx}\right)^{k+1}x^{k-r}H(x^r) &= -r\cdot k! \sum_{\ell=0}^\infty \binom{k+r\ell}{k} \frac{1}{\ell+1} x^{r\ell-1}
\end{align*}
and we are done.
\end{proof}


To show that the $(k+1)$-st derivative is a rational function in $x$, we consider the product with $(1-x^r)^k$ and show that this is a polynomial, which is not an obvious result. 
\begin{lem}\label{lemterminating}
    For integer $k \geq 1$ and real $0< x< 1$ we have 
    $$ \left(\frac{d}{dx}\right)^{k+1}x^{k-r}H(x^r) = -\frac{r\cdot k!}{x(1-x^r)^{k}} \sum_{j=0}^{k-1} x^{rj}\sum_{v=0}^j\frac{(-1)^{j-v}}{v+1}\binom{rv+k}{k} \binom{k}{j-v}.$$
\end{lem}
\begin{proof}
    For $0<  x < 1$, where the entropy series converges, reindex the product
\begin{align*}
    (1-x^r)^k\sum_{\ell=0}^\infty \binom{k+r\ell}{k} \frac{1}{\ell+1} x^{r\ell} &= \sum_{m=0}^k (-1)^m \binom{k}{m}x^{rm} \sum_{\ell=0}^\infty \binom{k+r\ell}{k} \frac{1}{\ell+1} x^{r\ell} \\
    &= \sum_{j=0}^\infty x^{rj} \sum_{v=0}^j \frac{(-1)^{j-v}}{v+1} \binom{k+rv}{k} \binom{k}{j-v} \\
    &= \sum_{j=0}^\infty x^{rj} h_{k,r,j}
\end{align*}
where we recall the definition of $h_{k,r,j}$ in Equation \eqref{hkrjdef}. Now Lemma \ref{thm:finite-diff} says that for $k\geq r\geq 1$ and $j\geq k$, we have $h_{k,r,j}=0$. Therefore this sum is actually a polynomial, and
$$\sum_{\ell=0}^\infty \binom{k+r\ell}{k} \frac{1}{\ell+1} x^{r\ell} = \frac{1}{(1-x^r)^k} \sum_{j=0}^{k-1} x^{rj} h_{k,r,j}.$$
Comparing with Lemma \ref{lem8} completes the proof.
\end{proof}


\begin{lem}\label{lem:stir1}
Let $r,n\geq 1$ be integers and $r\leq s\leq n+r-1$ an integer. For complex $w$ with $\Re(w)<1$ we have
\begin{equation}\label{lem:stir1eq1}
    \sum_{\ell=0}^\infty w^{r\ell-1} \binom{r\ell+s}{n} = \frac{1}{n!} \sum_{\ell=0}^n \ell! S(n,\ell|1,r,s)r^\ell \frac{w^{r\ell-1}}{(1-w^r)^{\ell+1}}
\end{equation}
and
\begin{equation}\label{lem:stir1eq2}
  \sum_{\ell=0}^\infty \frac{w^{r\ell-1}}{\ell+1} \binom{r\ell+s}{n} =  \frac{1}{n!} \sum_{\ell=0}^n \ell!S(n,\ell+1|1,r,s-r)r^{\ell+1} \frac{w^{r\ell-1}}{(1-w^r)^{\ell+1}}.
\end{equation}
\end{lem}
\begin{proof}
We begin with the Dobi\'nski-type formula of Equation \eqref{stirdef4}:
\begin{equation}
    \sum_{\ell=0}^\infty \frac{x^\ell}{\ell!} \binom{r\ell+s}{n} =  \frac{1}{n!} \sum_{\ell=0}^n S(n,\ell|1,r,s)r^\ell e^xx^\ell.
\end{equation}
We will take Laplace transforms of both sides. Note that the Laplace transform with $\Re(w)>1$ acts on monomials as
$$\int_0^\infty e^{-wx}x^\ell dx = \frac{\ell!}{w^{\ell+1}}, $$
so that 
$$\int_0^\infty e^{-wx}e^x x^\ell dx = \frac{\ell!}{(w-1)^{\ell+1}}. $$
Laplace transforming both sides gives
\begin{align}
    \sum_{\ell=0}^\infty\frac{1}{w^{\ell+1}} \binom{r\ell+s}{n} &=\frac{1}{n!} \sum_{\ell=0}^n S(n,\ell|1,r,s)r^\ell \int_0^\infty e^{(1-w)x} x^\ell dx  \\
    &=\frac{1}{n!} \sum_{\ell=0}^n S(n,\ell|1,r,s)r^\ell \frac{\ell!}{(w-1)^{\ell+1}} 
\end{align}
with $\Re(w)>1$. Now mapping $w\mapsto 1/w$ gives
$$ \sum_{\ell=0}^\infty{w^{\ell+1}} \binom{r\ell+s}{n} = \frac{1}{n!} \sum_{\ell=0}^n \ell! S(n,\ell|1,r,s)r^\ell \frac{w^{\ell+1}}{(1-w)^{\ell+1}} $$
with $\Re(w)<1$.
Dividing by $w$, mapping $w\mapsto w^r$, and dividing by $w$ again gives the first result.

For the second result with the $\frac{1}{\ell+1}$ factor, we again begin with 
\begin{equation*}
    \sum_{\ell=0}^\infty \frac{x^\ell}{\ell!} \binom{r\ell+s}{n} =  \frac{1}{n!} \sum_{\ell=0}^n S(n,\ell|1,r,s)r^\ell e^xx^\ell,
\end{equation*}
 separate out the $\ell=0$ terms on both sides, shift $\ell$ by $1$, and divide through by $x$:
 $$   \binom{s}{n} + \sum_{\ell=1}^\infty \frac{x^\ell}{\ell!} \binom{r\ell+s}{n} = \frac{1}{n!}S(n,0|1,r,s)e^x+ \frac{1}{n!} \sum_{\ell=1}^n S(n,\ell|1,r,s)r^\ell e^xx^\ell
$$
and
 $$   \frac{1}{x}\binom{s}{n} + \sum_{\ell=0}^\infty \frac{x^\ell}{(\ell+1)!} \binom{r\ell+r+s}{n} = \frac{1}{n!}S(n,0|1,r,s)\frac{e^x}{x}+ \frac{1}{n!} \sum_{\ell=0}^{n-1} S(n,\ell+1|1,r,s)r^{\ell+1} e^xx^{\ell}.
$$

Note that the Laplace transform of $1/x$ does not exist, so we need both of the initial terms to drop. When  $0\leq s<n$ is an integer, the binomial coefficient evaluates to $0$ and $S(n,0|1,r,s) = s(s-1)\cdots (s-n+1) = 0$. Now Laplace transform both sides with $\Re(w)>1$:
 $$  \sum_{\ell=0}^\infty \frac{1}{(\ell+1)w^{\ell+1}} \binom{r\ell+r+s}{n} =  \frac{1}{n!} \sum_{\ell=0}^{n-1} S(n,\ell+1|1,r,s)r^{\ell+1} \frac{\ell!}{(w-1)^{\ell+1}}.
$$
Map $w\mapsto 1/w$, so $\Re(w)<1$, divide by $w$, and map $s\mapsto s-r$ so that $r\leq s < n+r$:
\begin{equation}
\sum_{\ell=0}^\infty \frac{w^\ell}{\ell+1} \binom{r\ell+s}{n} =  \frac{1}{n!} \sum_{\ell=0}^{n-1} S(n,\ell+1|1,r,s-r)r^{\ell+1}\ell! \frac{w^\ell}{(1-w)^{\ell+1}}.    
\end{equation}
Map $w\mapsto w^r$ and divide by $w$:
 $$  \sum_{\ell=0}^\infty \frac{w^{r\ell-1}}{\ell+1} \binom{r\ell+s}{n} =  \frac{1}{n!} \sum_{\ell=0}^{n-1} \ell!S(n,\ell+1|1,r,s-r)r^{\ell+1} \frac{w^{r\ell-1}}{(1-w^r)^{\ell+1}}
$$
to finish.
\end{proof}

An alternate proof of Lemma \ref{lemterminating} proceeds by starting with Equation \eqref{lem:stir1eq2}, clearing denominators by $(1-w^r)^{n}$, using the binomial theorem on $(1-w^r)^{n-\ell}$, and inserting the closed form expression for generalized Stirling numbers from Equation \eqref{stirdef7}. Then we can switch the order of summation in the triple sum and evaluate the innermost sum using a classical binomial identity to get back down to a double sum.


Combining all of these lemmas proves Theorem \ref{thm3}, giving closed forms for the $(k+1)$-st derivative of $x^{k-r}H(x^r)$.

\section{Special cases}\label{sec4}
We will simplify the cases $r=1,k$ which Yuster originally studied in \cite{Yuster23}. This will prove Corollaries \ref{cor6} and \ref{cor7}. 
The following sequence, which has been studied many times, makes an appearance. 
\begin{defn}
Define the \textbf{$s$-binomial coefficients} through the generating function
\begin{equation}\label{snomialdef}
    \sum_{\ell=0}^{ks} \binom{k}{\ell}_s x^\ell := (1+x+x^2+\cdots + x^s)^k = \left(\frac{1-x^{s+1}}{1-x}\right)^k.
\end{equation}
\end{defn}
A 1731 result of de Moivre \cite{DM1731} gives the closed form
\begin{equation}\label{sbinom}
    \binom{k}{\ell}_{s-1} = \sum_{v=0}^{\left \lfloor \ell/s \right \rfloor} (-1)^v \binom{k}{v}\binom{\ell-vs+k-1}{k-1},
\end{equation}
where the restriction $v\leq \left \lfloor \ell/s \right \rfloor$ comes from setting $\ell-s+k-1\geq k-1$ so that the second binomial coefficient is positive.

We repeat the statement of Corollary \ref{cor7}. Equation \eqref{cor7eq2} proves an observation of Yuster that the coefficients are given by OEIS sequence A108267, which is $\binom{k}{\ell k}_{k-1} $.
\begin{cor}Consider real $0< x < 1$ and $\omega=e^{\frac{2\pi i}{k}} $ a primitive $k$-th root of unity. In terms of $s$-binomial coefficients defined in Definition \eqref{snomialdef}, 
\begin{align}
 \left(\frac{d}{dx}\right)^{k+1}H(x^k)  &= -k\cdot k! \sum_{\ell=0}^\infty \binom{k+k\ell-1}{k-1} x^{k\ell-1} \label{cor7eq3} \\
 &= -\frac{k!}{x} \sum_{j=0}^{k-1} \frac{1}{(1-\omega^jx)^k}\label{cor7eq1} \\
 &=-\frac{k\cdot k!}{x(1-x^k)^k} \sum_{\ell=0}^{k-1} \binom{k}{\ell k}_{k-1}x^{k\ell} \label{cor7eq2}.
 \end{align}
\end{cor}
\begin{proof}
Specializing Theorem \ref{thm3} to $r=k$ and noting the binomial coefficient identity
$$\binom{k+k\ell}{k} \frac{1}{\ell+1} =  \frac{k+k\ell}{k}\binom{k+k\ell-1}{k-1} \frac{1}{\ell+1} =\binom{k+k\ell-1}{k-1}$$
proves Equation \eqref{cor7eq3}. Now let $\omega = e^{2 \pi i/k}$ be a primitive $k$-th root of unity and write 
\begin{align*}
     \left(\frac{d}{dx}\right)^{k+1}H(x^k) &= -\frac{k\cdot k!}{x}\sum_{\ell=0}^\infty \binom{k+k\ell-1}{k\ell}x^{k\ell} \\
     &=-\frac{k\cdot k!}{x}\sum_{\ell=0}^\infty \binom{k+\ell-1}{\ell}x^{\ell} \mathbbm{1}\left[\ell \equiv 0 \pmod{k}\right] \\
     &= -\frac{k!}{x}\sum_{\ell=0}^\infty x^\ell \binom{k+\ell-1}{\ell}  \sum_{j=0}^{k-1} \omega^{j\ell}.
\end{align*}
We divided by $k$ since the inner sum along roots of unity is zero unless $\ell \equiv 0 \pmod{k}$, in which case it is $k$.  Now, we use the generalized binomial theorem to deduce Equation \eqref{cor7eq1}
\begin{align*}
    -\frac{k!}{x}\sum_{\ell=0}^\infty x^\ell \binom{k+\ell-1}{\ell}  \sum_{j=0}^{k-1} \omega^{j\ell} &= -\frac{k!}{x} \sum_{j=0}^{k-1} \sum_{\ell=0}^\infty \binom{k+\ell-1}{\ell} (\omega^j x)^\ell  = -\frac{k!}{x} \sum_{j=0}^{k-1} \frac{1}{(1-\omega^jx)^k}.
\end{align*}

Now note that if $F(z) = \sum_{n=0}^\infty a_n z^n$, we have $\sum_{n=0}^\infty a_{kn} z^{kn} = \frac{1}{k} \sum_{j=0}^{k-1}{F(w^jz)}$, where $\omega$ is a primitive $k$-th root of unity. Then by setting $F(x)= \left(\frac{1-x^k}{1-x}\right)^k$ to be the generating function of $\binom{k}{\ell}_{k-1}$, we have
\begin{align}
\sum_{\ell=0}^k \binom{k}{k\ell}_{k-1}x^{k\ell} = \frac{1}{k}\sum_{j=0}^{k-1} \left(\frac{1-(\omega^jx)^k}{1-\omega^jx}\right)^k = \frac{(1-x^k)^k}{k} \sum_{j=0}^{k-1} \frac{1}{(1-\omega^jx)^k},
\end{align}
which proves Equation \eqref{cor7eq2}. Note that this is a multisection identity which essentially computed the Fourier expansion of the $(k-1)$-binomial generating function.
\end{proof}

\begin{cor}
We have
\begin{equation}\label{eq:corbadlabel}
\left(\frac{d}{dx}\right)^{k+1}x^{k-1}H(x) = \frac{(k-1)!}{x^2} \left(1 - \frac{1}{(1-x)^k}\right). 
\end{equation}\end{cor}
\begin{proof}
Consider Equation \eqref{thm3eq1} with $r=1$, so that
$$ \left(\frac{d}{dx}\right)^{k+1}x^{k-1}H(x) = - k! \sum_{\ell=0}^\infty \binom{k+\ell}{k} \frac{1}{\ell+1} x^{\ell-1}.$$
Now use the generalized binomial theorem to show $$ \sum_{\ell=0}^\infty \binom{k+\ell}{k} \frac{1}{\ell+1} x^{\ell} =\frac{1}{k}\sum_{\ell=0}^\infty \binom{k+\ell}{\ell+1}  x^{\ell}  = \frac{1}{k}\sum_{\ell=1}^\infty \binom{k+\ell-1}{\ell}  x^{\ell-1} = \frac{1}{kx} \left(\frac{1}{(1-x)^{k}}-1\right),$$
and we are done.
\end{proof}

\section{Real rootedness reduction}\label{sec:reduction}

We finally show that Conjecture \ref{rootconj} about real roots implies inequality \eqref{mainineq} for real exponents. 
\begin{thm}
    The real rootedness Conjecture \ref{rootconj} implies the entropy inequality of Conjecture \ref{mainthm} for all real $k\geq 1$.
\end{thm}
\begin{proof}
Our proof follows the framework of \cite{Yuster23}, but with the extra parameter $r$. The flexibility given by the extra $r$ parameter is crucial to proving the reduction for real exponents, as opposed to integer exponents. Consider the function 
$$f_{k,r}(x):=\alpha H(x^k) - x^{k-r}H(x^r),$$
where $\alpha:=\alpha_{k/r}$ satisfies the function equation \eqref{alphafunc} with parameter $k/r$, which is equivalent to
\begin{equation}
    \alpha^r = \frac{1}{(1+\alpha)^{k-r}}.
\end{equation}
We omit the subscript in $\alpha_{k/r}$ for clarity. Our goal is to compute all roots of $f_{k,r}(x)$ in $[0,1]$. 

We have a trivial root at $x=1$ since $H(1)=0$. 

We have a double root at $\frac{1}{(1+\alpha)^{1/r}}$ since we can calculate that $f_{k,r}\left(\frac{1}{(1+\alpha)^{1/r}}\right)=f'_{k,r}\left(\frac{1}{(1+\alpha)^{1/r}}\right)=0$. Using the symmetry $H(x)=H(1-x)$ and the functional equation for $\alpha$, we have 
\begin{align*}
    f_{k,r}\left(\frac{1}{(1+\alpha)^{1/r}}\right) &= \alpha H\left( \frac{1}{(1+\alpha)^{k/r}}\right) - \frac{1}{(1+\alpha)^{k/r-1}}H\left(\frac{1}{1+\alpha}\right) \\
    &= \alpha H\left(\frac{\alpha}{1+\alpha}\right)-\alpha H\left(\frac{1}{1+\alpha}\right) \\
    &=0.
\end{align*}
We now compute the derivative
\begin{align}\label{fkrderiv}
    \frac{1}{x^{k-r-1}} \frac{d}{dx}f_{k,r}(x) = \alpha k x^{r}\log \left( \frac{1-x^k}{x^k}\right) -kx^r\log \left( \frac{1-x^r}{x^r}\right) + (k-r)\log(1-x^r).
\end{align}
Using the functional equation for $\alpha$ several times, at $x=\frac{1}{(1+\alpha)^{1/r}}$ we have 
$$1-x^r = \frac{\alpha}{1+\alpha}, \quad \frac{1-x^r}{x^r} =\alpha , \quad  \frac{1-x^k}{x^k} = (1+\alpha)^{k/r}-1 = \frac{1+\alpha}{\alpha}-1=\frac{1}{\alpha}.$$
Now note that $$ (k-r) \log \left(\frac{\alpha}{1+\alpha}\right) = k  \log \left(\frac{\alpha}{1+\alpha}\right) -  \log \left(\frac{\alpha}{1+\alpha}\right)^r =k  \log \frac{\alpha}{1+\alpha} -  \log \frac{1}{(1+\alpha)^k} = k\log \alpha. $$
Substituting this into the derivative \eqref{fkrderiv} gives 
\begin{align*}
    (1+\alpha)^{\frac{k-r-1}{r}}f'_{k,r}\left(\frac{1}{(1+\alpha)^{1/r}}\right) &= k \frac{\alpha }{1+\alpha} \log \frac{1}{\alpha} - \frac{k}{1+\alpha}\log \alpha + (k-r) \log \left(\frac{\alpha}{1+\alpha}\right) \\
    &= - k \frac{\alpha }{1+\alpha} \log \alpha - \frac{k}{1+\alpha}\log \alpha  + k\log \alpha \\
    &= 0.
\end{align*}

We also have a root of multiplicity $k$ at $x=0$. Equation \eqref{lem:fund2} states that
$$x^{k-r}H(x^r) = -rx^k\log x+x^k - \sum_{\ell=1}^\infty \frac{x^{k+r\ell}}{\ell(\ell+1)},$$
so that for $0\leq t\leq k-1$ we have 
\begin{align*}
    \left(\frac{d}{dx}\right)^{t}x^{k-r}H(x^r) \bigg\vert_{x=0} &= -r  \left(\frac{d}{dx}\right)^{t} x^k\log x \bigg\vert_{x=0}.
\end{align*}
Using the iterated product rule and separating the term at $\ell=0$ gives
\begin{align*}
    \left(\frac{d}{dx}\right)^{t} x^k\log x &= -r\sum_{\ell=0}^t\binom{t}{\ell} \left(\frac{d}{dx}\right)^{t-\ell} x^k \cdot \left(\frac{d}{dx}\right)^{\ell} \log x \\
    &= -r\frac{k!}{(k-t)!} x^{k-t}\log x  -r \sum_{\ell=1}^t \binom{t}{\ell}\frac{k!}{(k-t+\ell)!}x^{k-t+\ell} \frac{(-1)^{\ell-1} }{x^\ell} \\
    &= -r\frac{k!}{(k-t)!}x^{k-t}\log x  -r  \sum_{\ell=1}^t (-1)^{\ell-1} \binom{t}{\ell}\frac{k!}{(k-t+\ell)!}x^{k-t}.
\end{align*}
Irrespective of the value of $r$, for $0\leq t \leq k-1$ we have $\lim_{x\to 0}x^{k-t}\log x = 0$, which in turn means that 
$$ \left(\frac{d}{dx}\right)^{t} f_{k,r}(x) \bigg\vert_{x=0} =0.$$
Now, we appeal to Theorem \ref{thm3}, which states that 
\begin{align*}
    \left(\frac{d}{dx}\right)^{k+1}f_{k,r}(x) &= -\alpha\frac{k\cdot k! }{x(1-x^k)^{k}}h_{k,k}(x) + \frac{r\cdot k! }{x(1-x^r)^{k}} h_{k,r}(x) \\
    &= - \frac{k!}{x(1-x^r)^{k}(1-x^k)^{k}} \left(\alpha k (1-x^r)^{k} h_{k,k}(x) - r (1-x^k)^{k}h_{k,r}(x)\right),
\end{align*}
where $$ h_{k,r}(x)=\sum_{j=0}^{k-1} x^{rj}\sum_{v=0}^j\frac{(-1)^{j-v}}{v+1}\binom{rv+k}{k}  \binom{k}{j-v}$$ as before. Now assume the conjecture that the numerator has two real roots in $0<x<1$. By Rolle's theorem applied $k+1$ times to the $(k+1)$-st derivative, it follows that $f_{k,r}(x)$ contains at most $k+3$ roots in $[0,1]$, counting multiplicity. We have a trivial root at $x=1$, a double root at $x = \frac{1}{(1+\alpha)^{1/r}}$, and a root of multiplicity $k$ at $x=0$. Therefore, we have found all $k+3$ roots of $f_{k,r}(x)$ in $[0,1]$.

Because $f_{k,r}(x)$ has a double root at $\frac{1}{(1+\alpha)^{1/r}}$, and the other roots are at the endpoints of the interval $[0,1]$, it must be either non-positive or non-negative on $[0,1]$. Yuster \cite[Lemma 3.3]{Yuster23} showed that there is a small $\eps$ such that $f_{k,1}(x)>0$ for $0<x<\eps$ and integers $k\geq 2$. The exact same proof shows that there is an $\eps_r$ so that $f_{k,r}(x^{1/r})>0$ for $0<x<\eps_r$ and $k/r>1$. Since $f_{k,r}$ takes a positive value, it must be non-negative on $[0,1]$.

Given that $f_{k,r}(x)=\alpha_{k/r} H(x^k) - x^{k-r}H(x^r)\geq 0, 0\leq x\leq 1$, we now map $x\mapsto x^{1/r}$, which sends $[0,1]$ to $[0,1]$. Therefore $\alpha_{k/r}H\left(x^{k/r}\right) - x^{k/r-1}H(x)\geq 0$. However we picked $k>r\geq 1$ as arbitrary coprime integers, so that $k/r$ runs through all rationals greater than 1, and the inequality $\alpha_{q} H(x^q) - x^{q-1}H(x)\geq 0$ holds for all rational $q>1$. Since each term $\alpha_q, H(x^q), x^{q-1}$ is continuous in $q>1$, the inequality must also hold for all real $q>1$. The inequality is also trivial at $q=1$, which finishes the proof.
\end{proof}

The previous proof shows that if we can verify that $p_{k,r}(x):=\alpha_{k/r} k (1-x^r)^{k} h_{k,k}(x) - r (1-x^k)^{k}h_{k,r}(x)$ has two roots in $(0,1)$ for a fixed pair of integers $k,r$, then we have verified inequality \eqref{mainineq} for the rational exponent $k/r$. For instance, at $k=3,r=2,\alpha_{3/2}\approx 0.754878,$ this polynomial is
\begin{align*}
   p_{3,2}(x)&=3 \alpha_{3/2}\left(-x^{12}+3 x^{10}-7 x^9-3 x^8+21 x^7-21 x^5+3 x^4+7 x^3-3 x^2+1\right)  \\
   &\quad -\left(\frac{2 x^{13}}{3}-4 x^{11}-2 x^{10}-2 x^9+12 x^8+2 x^7+6 x^6-12
   x^5-\frac{2 x^4}{3}-6 x^3+4 x^2+2\right) \\
   &\approx -\frac{2x^{13}}{3} - 2.26x^{12} +4x^{11}+8.79x^{10}-13.85x^9-18.79x^8+45.56x^7-6x^6-35.56x^5+7.46x^4 \\
   &\quad+21.885x^3-10.79x^2+0.26,
\end{align*}
which has seven sign changes in the coefficients and is not suited to an application of Descartes' rule of signs. Instead, we can numerically evaluate that this has two real roots in $(0,1)$ at $\approx 0.204863, 0.74186,$ which proves the main entropy inequality \eqref{mainineq} for the fractional exponent $3/2$.

Alternatively, as noted in the introduction we could consider the transformed polynomial 
\begin{align*}
    (1+y)^{k^2+kr-r} p_{k,r}\left(\frac{1}{1+y}\right)  = (1+y)^{13} &p_{3,2}\left(\frac{1}{1+y}\right) \approx 0.26y^{13} + 3.44y^{12}+9.85y^{11}-21.20 y^{10} -178.47y^9\\
    &-425.46y^8-507.46y^7-309.02y^6-62.01y^5+32.79y^4+19.05y^3.
\end{align*}
The coefficients can be provably correctly computed to arbitrary accuracy using interval arithmetic, so we can read off that there are exactly two sign changes in the coefficients, which correspond to two real roots $y_1,y_2 \in (0,\infty)$ by Descartes' rule of signs. This then corresponds to two real roots of $p_{3,2}(x)$ in $(0,1)$.

Also note that we can factor $(1-x)^k$ out of $p_{k,r}(x)$, while still leaving a polynomial. Equivalently, we can show that 
$$\alpha_{k/r} k \left(\frac{1-x^r}{1-x}\right)^{k} h_{k,k}(x) - r \left(\frac{1-x^k}{1-x}\right)^{k}h_{k,r}(x)$$
has two real roots in $(0,1)$, counting multiplicity. The term $\left(\frac{1-x^r}{1-x}\right)^{k}$ is the generating function for $(r-1)$-binomial coefficients given in Definition \eqref{snomialdef}. The $r=1$ case of this factored polynomial is exactly the polynomial $p_k(x)$ of Yuster \cite[Corollary 3.7]{Yuster23} which arose in his study of inequality \eqref{mainineq} for integer $k$. The $k=2,r=1$ case is additionally the polynomial $p(x)$ of Boppana \cite{Boppana23}.

\section{Functional equation}\label{sec:alpha}
We now collect some useful properties of $\alpha_k$, including basic bounds and first order asymptotics. Recall that $\alpha_k$ satisfies the functional equation \eqref{alphafunc}
$$\alpha_k = \frac{1}{(1+\alpha_k)^{k-1}}.$$
Note that the following result is tight since $\lim_{k\to 1^+}\alpha_k = 1$. 
\begin{lem}\label{lem:alphabound}
    For real $k>1$, $\alpha_k$ monotonically decreases in $k$ and satisfies
    \begin{equation}
     \frac{1}{k}< \alpha_k <1.        
    \end{equation}
\end{lem}
\begin{proof}
    Consider the functional equation $x_k +x_k^k=1$, written in terms of $x_k = \frac{1}{1+\alpha_k}$. This is monotonic in $0<x_k<1$ so has a unique solution in $(0,1)$, which corresponds to a unique value of $\alpha_k$ in $(0,1)$ satisfying \eqref{alphafunc}. If $k$ increases, the power $0<x_k^k<1$ decreases, so $x_k$ must monotonically increase. Then $\alpha_k = \frac{1}{x_k}-1$ monotonically decreases. Noting that $\lim_{k\to 1^+}\alpha_k = 1$ gives the upper bound.

    Assume $\alpha_k\leq 1/k$, then $x_k = \frac{1}{1+\alpha_k} \geq \frac{k}{k+1}$. Then we apply Bernoulli's (strict) inequality to $x_k+x_k^k \geq \frac{k}{k+1} + \left(1-\frac{1}{k+1}\right)^k > \frac{k}{k+1} + \frac{1}{k+1}=1$, which contradicts the functional equation $x_k+x_k^k=1$ and gives the lower bound. 
\end{proof}
We can compute the large $k$ asymptotics of $\alpha_k$. Note that $b_k \approx \log \log k$ to first order, but there are multiplicative corrections of order $\frac{1}{\log k}, \frac{1}{\log^2 k}, \ldots$.
The point of making $b_k$ the solution to an exact equation is that the remaining error term in Lemma~\ref{lem:alphak-estimate} is much smaller.
\begin{lem}\label{lem:alphak-estimate}
Let $b_k$ be the unique solution to
\begin{equation}\label{eq:def-bk}
    b_k - \log \left(1 - \frac{b_k}{\log k}\right) = \log \log k.
\end{equation}
    In the large $k$ limit, we have
\begin{align}
    \alpha_k &= \frac{\log k - b_k}{k} + O\left(\frac{\log^2 k}{k^2}\right) \\
    &= \frac{\log k}{k} + O\left( \frac{\log \log k}{k}\right).
\end{align}
\end{lem}
\begin{proof}
    We will do our calculations in $x_k = \frac{1}{1+\alpha_k}$, which is the unique solution of $x_k + x_k^k = 1$.
    
    We will guess for now that $x_k = 1 - \frac{\log k - \delta}{k}$ for $\delta \in [0,2\log \log k]$. We will see below that there is a solution $x_k$ of this form, which must be the unique solution.
    We calculate
    \begin{align*}
        \log x_k &= \log\left( 1 - \frac{\log k - \delta}{k}\right)  = -\frac{\log k - \delta}{k} + O\left(\frac{\log^2 k}{k^2}\right), \\
         \log x_k^k &= \delta -\log k + O\left(\frac{\log^2 k}{k}\right).
    \end{align*}
    Moreover $$\log(1-x_k)
        = \log(\log k - \delta) - \log k
        =  \log \log k + \log\left(1 - \frac{\delta}{\log k}\right) - \log k.$$
    The equation $x_k + x_k^k = 1$ implies $\log x_k^k = \log (1-x_k)$, so
    $$\delta - \log k + O\left(\frac{\log^2 k}{k}\right) = \log \log k + \log\left(1 - \frac{\delta}{\log k}\right) - \log k,$$
    which rearranges to
    $$\delta - \log\left(1 - \frac{\delta}{\log k}\right) = \log \log k + O\left(\frac{\log^2 k}{k}\right).$$
    This equation has a solution $\delta \in [0, 2 \log \log k]$ by the intermediate value theorem, and by inspection $\delta = b_k + O(\log^2 k / k)$.
    Therefore $$x_k = 1 - \frac{\log k - b_k + O(\log^2 k / k)}{k},$$ which implies the estimate on $\alpha_k = \frac{1-x_k}{x_k}$.
\end{proof}

Finally, we can give a series expansion for $x_k$ using Lagrange inversion. Note that the lower index of the binomial coefficient is $j$, as opposed to the $kj$ which appears in the definition of $h_{k,r}(x)$.
\begin{lem}\label{lem:lagrange}
    We have the following series expansion for $\alpha_k$:
    \begin{equation}
        x_k^N = \frac{1}{(1+\alpha_k)^N} = \sum_{j=0}^\infty (-1)^j\frac{N}{(k-1)j+N}\binom{kj+N-1}{j}.
    \end{equation}
\end{lem}
\begin{proof}
    Rewrite the functional equation $x_k+x_k^k=1$ as $x_k = \frac{1}{1+x_k^{k-1}}$. Consider $x_k(z)$ given as the solution of $$ x_k(z) = \frac{z}{1+x_k(z)^{k-1}}.$$
    We now perform Lagrange inversion along the variable $z$ in $x_k(z)$ before setting $z=1$, following \cite[Equation (2.2.1)]{Gessel16}. We have 
    $$[z^n] x_k(z)^N = \frac{N}{n}\left[t^{n-N}\right] \frac{1}{(1+t^{k-1})^n}  = \frac{N}{n}\left[t^{n-N}\right] \sum_{ j=0}^\infty \binom{n-1+j}{j}(-1)^j t^{(k-1)j}.$$
    The inner coefficient is only nonzero when $n-N = (k-1)j$, or when $n= (k-1)j+N$ for some $j$. Therefore 
    \begin{align*}
        x_k(z)^N &= \sum_{n=0}^\infty z^n \frac{N}{n}\left[t^{n-N}\right] \frac{1}{(1+t^{k-1})^n}\\
        &= \sum_{j=0}^\infty z^{(k-1)j+N}  \frac{N}{(k-1)j+N} \binom{kj+N-1}{j}(-1)^j.
    \end{align*}
    Now setting $z=1$ recovers $x_k^N$.
\end{proof}

\section{Acknowledgements}
We thanks Brice Huang for his proof of the precise asymptotics of $\alpha_k$ and Christian Krattenthaler for his introduction of finite difference operators. As always, we thank Christophe Vignat for helpful discussions.

\bibliographystyle{amsalpha}
\bibliography{bibliography}

\end{document}